\documentclass[conference]{IEEEtran}
\usepackage{amsmath,amssymb,amsthm,dsfont,stfloats,color,subfigure,cite}
\usepackage{hyperref}
\usepackage[pdftex]{graphicx}

\newtheorem{thm}{Theorem}

\newtheorem{lem}{Lemma}
\newtheorem{corol}{Corollary}

\def\argmax{\mathop{\rm argmax}}

\def\snr{{\small\textsf{SNR}}}

\def\inr{{\small\textsf{INR}}}

\long\def\comment#1{}

\newfont{\bbb}{msbm10 scaled 700}

\newfont{\bbc}{msbm10 scaled 1100}

\newcommand{\Cc}{{\cal C}}

\newcommand{\Mc}{{\cal M}}

\newcommand{\Oc}{{\cal O}}

\newcommand{\Qc}{{\cal Q}}

\newcommand{\Sc}{{\cal S}}

\newcommand{\Sigmam}{\hbox{\boldmath$\Sigma$}}

\newcommand{\trace}{{\hbox{tr}}}

\def\pc{{\pmb c}}\def\ph{{\pmb h}}\def\pn{{\pmb n}}\def\pw{{\pmb w}}\def\px{{\pmb x}}\def\py{{\pmb y}}\def\pz{{\pmb z}}

\def\pC{{\pmb C}}\def\pH{{\pmb H}}\def\pI{{\pmb I}}\def\pR{{\pmb R}}\def\p0{{\pmb 0}}

\def\ct{\mathsf{H}}

\IEEEoverridecommandlockouts

\begin{document}

\title{A Codebook-Based Limited Feedback System\\ for Large-Scale MIMO}

\author{\authorblockN{Junyoung Nam}\\
\authorblockA{ Wireless Communications Division\\ Electronics and Telecommunications Research Institute (ETRI), Daejeon, Korea } 
}

\maketitle

\begin{abstract}
In this paper, we consider limited feedback systems for FDD large-scale (massive) MIMO. A new codebook-based framework for multiuser (MU) MIMO downlink systems is introduced and then compared with an ideal non-codebook based system. We are particularly interested in the less-known finite-rate feedback regime where the number $M$ of transmit antennas and the number of users are of the same order of magnitude and $M$ is large but finite, which is a typical scenario of large-scale MIMO. We provide  new findings in this regime and identify some benefits of the new framework in terms of scheduling gain and downlink dedicated pilot overhead.

\end{abstract}

\section{Introduction}
\label{sec:intro}

Limited feedback MU-MIMO downlink systems may be categorized according to the use of a common codebook (CB). One is non-CB-based precoding with a direct (analog) quantization of channel vectors like zero-forcing beamforming (ZFBF) and the other is CB-based precoding, the latter of which refers to an opportunistic beamforming scheme like random beamforming (RBF) \cite{Sha05} where each user terminal (user) selects the channel direction information (CDI) with the maximum signal-to-interference-plus-noise ratio (SINR) and reports the CDI and SINR to the basestation (BS). By the CB-based scheme also employed in some practical systems (e.g., \cite{sesia2009lte}), we stipulate that the BS uses precoding vectors only in a common CB also known to all users.



In the past few years, we have paid great attention to large-scale MIMO \cite{Mar10}. A typical scenario therein is that the number of antennas, $M$, at the BS is larger than the number of users, $K$. 
Since the study on MU-MIMO systems has traditionally focused on the $M\ll K$ case,
new insights may be needed taking into account the much lesser known case where $M$ and $K$ have the same order of magnitude (denoted by $M\sim K$) with $M$ \emph{not very large}. For frequency division duplex (FDD) large-scale MIMO systems, limited feedback schemes based on noncoherent trellis-coded quantization \cite{Cho13} or compressive sensing \cite{Kuo12} were introduced, where the former is to reduce encoding complexity of channel quantization. More practical CB-based approaches have been considered not suitable for large-scale MIMO due to severe performance degradation. 
The goal of this work is not just to reduce CSI feedback overhead but rather to improve the performance of CB-based systems at the cost of a moderate feedback increase compared to the conventional CB-based approach.

A main issue for our purpose is as to whether to achieve a full scheduling gain. 
While ZFBF is a popular sub-optimal linear precoding scheme, its CDI feedback overhead (i.e., quantization bits) should grow linearly with $M\log_2\text{SNR}$ to achieve the full degrees of freedom \cite{Jin06}, i.e., unless the number of quantization bits is sufficiently large, the non-CB-based approach with ZFBF is subject to severe performance degradation. Also, it is quite sensitive to noisy channel state information at the transmitter (CSIT) due to channel estimation errors. However, ZFBF can provide relatively accurate SINRs (relatively tight lower bound on the exact value) of any selected users \cite{Yoo07}, which allows \emph{flexible scheduling} at the BS. The capability of flexible scheduling is very important unless $M\ll K$. If $M\sim K$, the probability that the scheduler selects $\min\{M,K\}$ users is quite low irrespectively of ZFBF or even the capacity-achieving dirty paper precoding \cite{Cai03} unless the SNR is sufficiently high. Thus, the BS should be able to estimate a variety of SINRs per user with the number $S$ of selected users less than $M$.   
Assuming perfect CSI at receivers, the CB-based approach allows users to calculate the exact SINRs for any given subset of beams in CB.  However, it requires each user to report an individual SINR per subset, thus making flexible scheduling infeasible as the number of all possible subsets becomes extremely large unless $M$ is very small. This limited scheduling explains to some extent why a CB-based system suffers significant degradation for $M$ large.

The second but more important issue for $M$ large is the cost of downlink dedicated (user-specific) pilot for coherent detection. 
For sufficiently large $M$ such that there is no interference between users at all, downlink dedicated pilot is unnecessary even for the simple maximal-ratio transmission (MRT) scheme \cite{Mar10} with perfect CSIT, since the $\ell^2$ norm of user channel vectors converges to a real positive deterministic value due to the channel hardening effect \cite{Hoc04} in the i.i.d. Rayleigh fading channel. However, if inter-user interference is somehow present, the use of orthogonal dedicated pilots would be essential for the system performance regardless of FDD or time division duplex (TDD) and of how $\min\{M,K\}$ is large. For instance, ZFBF suffers from its high sensitivity to imperfect (noisy or outdated) CSIT \cite{Yoo06}, which causes significant interference as $S$ increases for $M$ large. Therefore, assuming imperfect CSIT, we need orthogonal downlink dedicated pilots even in ZFBF as well as MRT and codebook-based schemes. As $\min\{M,K\}$ increases, the number $S$ of selected users will also grow and then the downlink dedicated pilots consume more resources, which ironically reduces the remaining resources for data streams, thus rather degrading the system performance. 
As a consequence, the cost of downlink dedicated pilot becomes a significant bottleneck to the performance of large-scale MIMO systems with \emph{imperfect CSIT}
The resulting performance limit of large-scale MIMO was characterized in \cite{Nam14a}.

In order to address the above issues for the large-scale array regime, we propose a new framework for CB-based MU-MIMO systems. The proposed framework takes advantage  of a new type of channel quality information (CQI) feedback. The framework is shown to enable flexible scheduling with moderate CQI feedback overhead and to significantly reduce the cost for dedicated pilot. 

\emph{Notations:} $A\sim B$ means the quantities $A$ and $B$ have the same order of magnitude.  





\section{Preliminaries}
\label{sec:Pre}

\subsection{System Model}

We consider a single cell MIMO downlink channel where the BS has $M$ transmit antennas and $K$ users have a single antenna each.
Let ${\pH}$ denote the $M\times K$ system channel matrix given by
stacking the $K$ users channel vectors $\ph$ by columns. The signal vector received by the users is given by
\begin{align} \label{eq:Pre-1}
   \py={\pH}^\mathsf{H}\px +\pz
\end{align}
where $\px$ is the input signal with $S$ the rank of the input covariance
$\boldsymbol{\Sigma}=\mathbb{E}[\px\px^\ct]$ (i,e., the total number of independent data streams) and $\pz  \sim\mathcal{CN} (\p0, \pI)$ is the Gaussian noise (plus other-cell interference in multi-cell environment) at the receivers. The system has the total power constraint such that $\trace(\Sigmam)\le P$, where $P$ implies the total transmit SNR. Denoting by $p_k$, $\pw_k$, and $d_k$ the power, linear precoding vector, and information symbol of user $k$, respectively, the signal received by the selected user $k$ can be represented as
\begin{align} \label{eq:Pre-2}
   y_k=\sqrt{p_k}\ph_k^\ct\pw_k d_k +\sum_{j\in \Sc,j\neq k}\sqrt{p_j}\ph_k^\ct\pw_j d_j +z_k, \ k \in \Sc 
\end{align}
where $\Sc \subseteq[1:K]$ is the set of users selected by the BS scheduler such that $|\Sc|=S\le \min(M,K)$ and $\sum_{k\in \Sc} p_k\le P$. 
The SINR of user $k$ is given by
\begin{align} \label{eq:Pre-4}
   \gamma_k = \frac{|p_k\ph_k^\ct\pw_k|^2}{\sigma_k^2+\sum_{j\in \Sc,j\neq k}|{p_j}\ph_k^\ct\pw_j|^2}.
\end{align}

A popular limited feedback system model is based on the user feedback in terms of CDI and CQI. While CDI of user $k$ conveys CSI on the channel direction $\ph_k/\|\ph_k\|$, CQI is often represented by the received SINR or the norm of $\ph_k$.

For ZFBF with perfect CSIT, the CDI of user $k$ becomes the channel vector $\ph_k$ itself and the BS needs no CQI feedback but $\sigma_k^2$ (i.e., noise plus other-cell interference power) of users for rate adaptation of each user. In this case, the BS can calculate the exact SINR in (\ref{eq:Pre-4}) of user $k$. We assume no quantization for ZFBF in this paper to provide a more clean view, even though the comparison with a limited feedback scheme is unfair. 
Without perfect CSIT, we often resort to a CB-based scheme in practice, as will be reviewed in the sequel.

\subsection{Typical Framework for CB-Based MU-MIMO Systems}

\subsubsection{CDI Feedback}

For convenience, a CB $\boldsymbol{\Cc}$ of size $MT$ based on unitary beamforming is considered for the CDI quantization. 
One papular way is to use the well-known discrete Fourier transform (DFT) based approach already employed in real-world systems like LTE as follows:
The CB is constructed by truncating the first $M$ rows of the $MT$-point DFT matrix and by normalizing to satisfy the unit-norm property of precoding vector $\pw_k\in\boldsymbol{\Cc}$ for all $k$. The resulting precoding vector $\pc_\ell$ is given by 
\begin{align} \label{eq:Pre-3}
   \pc_\ell = \frac{1}{\sqrt{M}}\Big[1 \ e^{-j2\pi \ell\frac{1}{MT}} \ e^{-j2\pi \ell\frac{2}{MT}} \cdots \ e^{-j2\pi \ell\frac{M-1}{MT}}\Big]^T 
\end{align}
where $\ell = \ell_2T + \ell_1$ with $\ell_1=0,1,\cdots, T-1$ and $\ell_2=0,1,\cdots,M-1$, thus yielding the expression of $\ell\triangleq(\ell_1,\ell_2)$. We also denote by 
\begin{align} \label{eq:Pre-3b}
   \pC_t =\{\pc_\ell: \ell_1=t, \forall \ell_2\}  
\end{align} 
a subset of $\boldsymbol{\Cc}$ which forms a unitary matrix. For notational convenience we will use the notation $\pc_m^{(t)}$, where $t=\ell_1$ and $m=\ell_2$, instead of $\pc_\ell$, if necessary. 

We follow a similar way to \cite{sesia2009lte, Yan10}. Assuming all $M$ beams for each $\pC_t$ are scheduled with the equal power allocation, i.e., $p_k=\frac{P}{M}$, user $k$ selects the best beam index $\ell_k= (t_k,m_k)$ that maximizes the corresponding SINR such that   
\begin{align} \label{eq:Pre-5}
   (t_k,m_k) = \argmax_{t,m} \gamma_k(t,m).
\end{align}
where $$\gamma_k(t,m) = \frac{\big|\ph_k^\ct\pc_m^{(t)}\big|^2}{\frac{M}{P}\sigma_k^2+\sum_{j\in [1:M],j\neq m}\big|\ph_k^\ct\pc_j^{(t)}\big|^2}.$$
Each user is assumed to report its CDI as a $\lceil\log_2 MT\rceil$-bit index $\ell_k$ through error-free and delay-free uplink to the BS.

\subsubsection{CQI Feedback}

The CQI feedback is intended for the BS to estimate the SINRs experienced by users for scheduling and rate adaptation. An accurate estimate of the SINRs is generally complicated in the limited feedback case as the exact interference cannot be evaluated at users without a priori information on the scheduling decision. Notice that the CDI in (\ref{eq:Pre-5}) is just the result on a best-effort basis at the individual user side. Therefore, scheduling with such CDI and CQI is restrictive especially when $K$ is small, thus yielding performance degradation.  In general each user is required to report either the SINR $\gamma_k(\ell_k)$ corresponding to $\ell_k$ or its achievable rate $\log_2(1+\gamma_k(\ell_k)).$ 


\subsubsection{Scheduling}

A user selection algorithm can be given as follows: Let $\Sc_m^{(t)} =\{k: \ell_k=(t,m)\}$ and $\Sc^{(t)}=\emptyset$.
For each $(t, m)$, select the best user such that 
\begin{align} \label{eq:Pre-6}
   k^* &= \argmax_{k\in \Sc_m^{(t)}} \log_2(1+\gamma_k(\ell_k)) \nonumber \\
   \Sc^{(t)} &= \Sc^{(t)} \cup \{k^*\}.
\end{align}
Then, select the best $t$ such that
\begin{align} \label{eq:Pre-7}
   t^* &= \argmax_t \sum_{k\in \Sc^{(t)}} \log_2(1+\gamma_k(\ell_k)) \nonumber \\
   \Sc &=  \Sc^{(t^*)}.
\end{align}

It is easy to see that the above scheduling is restrictive due to limited CQI feedback.
Moreover, while CQI is fed back assuming all $M$ beams are scheduled, the number $S$ of selected users is often less than $M$, thus implying the power allocation of $\frac{P}{M}$ does not make sense but users have no \emph{a priori} information on this. 
Even if we restrict to schedule beams less than $M$, the BS hardly knows which beams are better to select and how many beans are good enough. 
As a result, the scheduler for DFT-SINR makes use of (\ref{eq:Pre-5}) as a lower bound on the exact SINR $\gamma_k$ in (\ref{eq:Pre-4}).




We refer to the above scheme as DFT-SINR in this paper as it uses the DFT-based beamforming with SINR feedback.

\section{New Framework based on INR Feedback}
\label{sec:NF}

In the typical framework, each user is required to feed back the maximum SINR in (\ref{eq:Pre-5}), assuming that the BS schedules and transmits all the $M$ beams for each $\pC_t$ and that all users have the equal transmit power $P/M$. Although these assumptions are inevitable due to feedback overhead, they do not make much sense in realistic systems unless $K\gg M$, which is even more infeasible for $M$ large. Therefore, we would like the scheduler to select users in a fully flexible manner over all possible $s$-combinations of $[1:MT]$ for $s=1,\cdots, M$, whose size is $\sum_{s=1}^M C^{MT}_s$, where $C^{MT}_s= \tbinom{MT}{ s}$.  Denote by $\Mc_s$ the set of each $s$-combinations of $[1:MT]$ and by $\Mc_{s,n}$ its $n$th element (it is also a set of size $s$), where $n=1,2,\cdots,C^{MT}_s$. Assuming the equal power allocation again such that $p_k=\frac{P}{s}$, the SINRs of user $k$ with respect to precoding vectors in $\Mc_{s,n}$ can be expressed as  
\begin{align} \label{eq:NF-1}
   \gamma_k(s,n,q) = \frac{\big|\ph_k^\ct\pc_q\big|^2}{\frac{s}{P}\sigma_k^2+\sum_{j\in \Mc_{s,n}\setminus q}\big|\ph_k^\ct\pc_j\big|^2}, \ q \in \Mc_{s,n}.
\end{align}

Although the above SINRs of all users allow the BS to optimally select users for a given CB $\boldsymbol{\Cc}$, the overhead of $\sum_{s=1}^M C^{MT}_s$ CQIs per user is prohibitively large even for $M$ small. 
In order to overcome this difficulty, we propose a new framework based on the following CQI feedback.

\subsection{INR Feedback}
\label{sec:NF-A}

For each beam $\pc_\ell$ in $\boldsymbol{\Cc}$, we define the corresponding interference to noise ratio (INR) as
\begin{align} \label{eq:NF-3}
   \inr_{k,\ell} &= \frac{\big|\ph_{k}^\ct\pc_\ell\big|^2}{\sigma_{k}^2}, \ \ell \in [1:MT].
\end{align}
If the above INRs are available at the BS, the scheduler can select one of them as the numerator of SINR and then calculate the corresponding all possible SINRs in (\ref{eq:NF-1}) for user $k$, as easily shown by
\begin{align} \label{eq:NF-5}
   \frac{\big|\ph_k^\ct\pc_q\big|^2}{\frac{s}{P}\sigma_k^2+\sum_{j\in \Mc_{s,n}\setminus q}\big|\ph_k^\ct\pc_j\big|^2} 
   &=    \frac{\frac{\big|\ph_k^\ct\pc_q\big|^2}{\sigma_k^2}}{\frac{s}{P}+\sum_{j\in \Mc_{s,n}\setminus q}\frac{\big|\ph_k^\ct\pc_j\big|^2}{\sigma_k^2}} \nonumber \\
   &= \frac{\inr_{k,q}}{\frac{s}{P}+\sum_{j\in \Mc_{s,n}\setminus q}\inr_{k,j}}.
\end{align}
Therefore, the BS needs only $MT$ INRs per user to accurately estimate all possible $\sum_{s=1}^M C^{MT}_s$ SINRs per user for a given $\Cc$.

This new type of CQI feedback indeed enables \emph{fully flexible scheduling} for a given $\boldsymbol{\Cc}$. It is easy to see that the INR feedback in (\ref{eq:NF-3}) also suggest \emph{no need for any CDI feedback} on beam index $\ell_k$ like (\ref{eq:Pre-5}). 
The above INR feedback is referred to as the \emph{full INR feedback} as it allows fully flexible scheduling. However, $MT$ INRs per user may incur too much overhead especially for $M$ large, even if small $T$ (e.g., 2)  is generally suitable for CB-based MU-MIMO schemes. 

In order to reduce the feedback overhead, we propose another INR feedback scheme that  allows the BS to select users only with the same subset $\pC_t \subseteq \boldsymbol{\Cc}$ in (\ref{eq:Pre-3b}). We call it the \emph{partial INR feedback}. Selecting $t_k$ such that 
$$t_k= \argmax_{t\in [1:T]} \max_{m\in [1:M]}\big|\ph_{k}^\ct\pc_m^{(t)}\big|^2$$ 
user $k$ feeds back 
\begin{align} \label{eq:NF-4}
   \inr_{k,m}^{(t_k)} &= \frac{\big|\ph_{k}^\ct\pc_m^{(t_k)}\big|^2}{\sigma_{k}^2}, \ m \in [1:M]
\end{align}
along with the selected subset index $t_k$.
Thus the partial INR scheme entails the feedback overhead of $M$ INRs and one $\lceil\log_2 T\rceil$-bit index $t_k$ per user. 
Denoting by $\Mc_s^{(t_k)}$ the set of all $s$-combinations of $[1:M]$ (i.e., the power set of $[1:M]$) and by $\Mc_{s,n}^{(t_k)}$ its $n$th element with $|\Mc_{s,n}^{(t_k)}|=s$, where $n=1,2,\cdots,C^{M}_{s}$, the BS with this partial INR feedback from each user can calculate the following SINRs
\begin{align} \label{eq:NF-2}
   \gamma_k(t_k,s,n,q) = \frac{\big|\ph_k^\ct\pc_q^{(t_k)}\big|^2}{\frac{s}{P}\sigma_k^2+\sum_{j\in \Mc_{s,n}^{(t_k)}\setminus q}\big|\ph_k^\ct\pc_j^{(t_k)}\big|^2}, \ q \in \Mc_{s,n}^{(t_k)}
\end{align}
for all $k, s,n,$ and $q$.

\subsection{Benefits of INR Feedback}
\label{sec:NF-B}

\subsubsection{Flexible Scheduling}

The fully flexible scheduling enabled by the full INR feedback (\ref{eq:NF-3}) is to select users over all possible combinations with SINRs $\gamma_k(s,n,q)$ in (\ref{eq:NF-1}) for all $k,s,n,$ and $q$. Clearly, this is the optimal scheduling for a given CB $\boldsymbol{\Cc}$ within the typical framework in Section \ref{sec:Pre}. 
In contrast, the scheduling with the partial INR feedback (\ref{eq:NF-4}) and SINRs in (\ref{eq:NF-2}) is naturally under the constraint of scheduling users only with the same $\pC_t$ such that their precoding vectors are orthogonal to each other. 

A simple algorithm for the partial INR feedback can be given as follows.
We need to eventually find the set $\Sc$ of selected users and the set $\Qc$ of their precoding vectors in $\boldsymbol{\Cc}$. 

Initializations: $\Sc_0^{(t)} =\{k: t_k=t\}$ and $\Sc_{s,n}^{(t)}=\emptyset$, $\Qc_{s,n}^{(t)}=\emptyset$, $\mu_{s,n}^{(t)}=0$ for all $(t,s,n)$.

Step 1) For each $(t,s,n)$, where $t=1,\cdots,T$, $s=1,\cdots,M$, and $n=1,\cdots,C^{M}_{s}$, calculate the sum rate $\mu_{s,n}^{(t)}$ as the following loop with $i=1$: 
\begin{align} \label{eq:NF-6}
   q &= \Mc_{s,n}^{(t)}(i) \\
   k^* &= \argmax_{k\in \Sc_0^{(t)}} \log_2(1+\gamma_k(t,s,n,q)).  
\end{align}
If  $\Sc_{s,n}^{(t)} \cap \{k^*\} = \emptyset$ and $i < s$, then 
\begin{align} \label{eq:NF-7}
   \Sc_{s,n}^{(t)} &= \Sc_{s,n}^{(t)} \cup \{k^*\} \\
   \Qc_{s,n}^{(t)} &= \Qc_{s,n}^{(t)} \cup \{q\} \\
   \mu_{s,n}^{(t)} &= \mu_{s,n}^{(t)} + \log_2(1+\gamma_{k^*}(t,s,n,q)) \\
   i &= i+1 .
\end{align}
Otherwise, finish this loop and continue Step 1) until we get $\Sc_{s,n}^{(t)}, \Qc_{s,n}^{(t)},$ and $\mu_{s,n}^{(t)}$ for all $(t,s,n)$.

Step 2) Select $\Sc$ and $\Qc$ as follows:
\begin{align} 
   (t^*,s^*,n^*) &= \argmax_{t,s,n} \mu_{s,n}^{(t)} \label{eq:NF-8} \\
   \Sc &=  \Sc_{s^*,n^*}^{(t^*)} \\
   \Qc &=  \Qc_{s^*,n^*}^{(t^*)}.
\end{align}

The condition in Step 1) is to prevent allocating multiple beams to a single user and for simplicity here we ignored finding a best alternative user for such a beam. 
Similarly, the fully flexible (optimal) scheduling algorithm can be given with $\gamma_k(s,n,q)$ in (\ref{eq:NF-1}) and taking into account finding a best alternative user for the multiple beams case.


\subsubsection{Dedicated Pilot Overhead Reduction}

Apart from flexible scheduling, a probably more important benefit of the proposed INR feedback is that it can reduce the dedicated pilot overhead. Note that the INR gives us a side information as to mutual interference between users. Depending on the INR feedback, the BS scheduler can figure out how much the intended user $k$ with beam $\pw_k$ is interfered by the other user using beam $\pw_j\in \boldsymbol{\Cc}$ with $j\neq k$. Thus the scheduler can allocate a common resource to multiple dedicated pilots of users with their beams sufficiently low interfering each other.  
In sum, the INR feedback provides not only flexible scheduling but it can also significantly reduce the overhead for dedicated pilot. 

\subsubsection{Arbitrary Power Allocation} 

Although we have assumed the equal power allocation to calculate SINRs in (\ref{eq:NF-1}) and (\ref{eq:NF-2}) for the proposed framework, in fact the INR feedback allows unequal power allocation of the total power $P$ over any subset of users. Although the resulting power optimization problem is non-convex, an appropriate sub-optimal power allocation algorithm (e.g., water-filling) will bring a noticeable performance improvement especially when users have different SNRs. 
Furthermore, we may jointly optimize user selection and power allocation. However, developing power allocation algorithms is for further study.

\subsection{Feedback Overhead of INR}

The feedback overhead of the full INR feedback is $MT$ INRs (positive real values) per user, while that of the partial feedback is $M$ INRs and $\lceil\log_2 T\rceil$-bit subset index $t_k$ per user. In contrast, DFT-SINR requires only a single SINR and $\lceil\log_2 MT\rceil$-bit beam index $\ell_k$ per user. However, the latter shows non-negligible performance degradation due to lack of scheduling flexibility (i.e., multiuser diversity gain reduction) as $M$ increases and hence the ratio $K/M$ (the number of users per beam) decreases for $K$ fixed. 
Notice that the quantized feedback of SINR generally requires only a small number of bits per SINR in practice (e.g., 2 or 4 bits in \cite{LTE12}), which is also the case with INR. For instance, the partial INR feedback may consume only $(2M+1)$ bits per user for INR and $t_k$, lending itself to large-scale MIMO unless $M$ is too large.

\section{Throughput of INR-Based Schemes}
\label{sec:TA}


In this section, we show how large a flexible scheduling gain owing to the INR feedback is in the i.i.d. Rayleigh fading channel. We focus on the $M\sim K$ case and for simplicity let $T=1$  so that the full and partial feedback schemes become equivalent. 
It is well known that the achievable throughput of RBF is \cite{Sha05}
\begin{align} \label{eq:TA-1}
   \mathcal{R}_\text{RBF} &= \mathbb{E} \bigg[\sum_{m =1}^M \log\Big(1+\max_{1\le k\le K}\gamma_{k}(m)\Big)\bigg] +o(1) \nonumber\\ 
   &\approx M\mathbb{E} \bigg[ \log\Big(1+\max_{1\le k\le K}\gamma_{k}(m)\Big)\bigg] 
\end{align}
where $o(1)$ takes into account the fact that user $k$ may be the strongest user of two or more beams with a vanishing probability as $K$ goes to infinity. The equality follows from the fact that $\gamma_{k}(m)$ is i.i.d. over both $k$ and $m$ since the beam vectors of RBF forms a unitary matrix. For sufficiently large $K$, we have
\begin{align} \label{eq:TA-1b}
   \mathcal{R}_\text{RBF} &= M\log\log K +M\log \frac{P}{M} +o(1). 
\end{align}

Assuming that $\pc_1, \cdots, \pc_M$ are random orthonormal vectors like RBF and that all users have the same SNR with $\sigma_k^2=1, \forall k$, the SINR in (\ref{eq:NF-3}) can be rewritten as
\begin{align} \label{eq:TA-3}
   \gamma_k(s,n,q) = \frac{v}{\frac{s}{P}+y} , \ q \in \Mc_{s,n}
\end{align}
where $v=|\ph_k^\ct\pc_q|^2$ and $y=\sum_{j\in \Mc_{s,n}\setminus q}|\ph_k^\ct\pc_j|^2$. Since $\{\pc_q: q \in \Mc_{s,n}\}$ are orthonormal, $v$ is i.i.d. over $k$ and $q$ with $\chi^2(2)$ distribution and $y$ is $\chi^2(2s-2)$ distributed, thus resulting in that $\gamma_k(s,n,q)$ is i.i.d. over $k$. The distribution of such $\gamma_k(s,n,q)$ is well known to satisfy the sufficient condition in \cite{Uzg54} (also \cite[Corollary A.1]{Sha05} for extreme value theory.

Given $s$ and $q \in [1:M]$, it is easy to see that 
the number of occurrences of $q$ in $\Mc_{s}$ is $C^{M-1}_{s-1}$ due to the flexible scheduling over the $s$-combinations of $[1:M]$. Denote by $\Mc_{s}(q)$ such a subset of $\Mc_{s}$ for $q$.  As a consequence, $\gamma_{k}(s,n',q)$ is i.i.d. over $n'\in \Mc_{s}(q)$ as well by the same argument as above.

The flexible scheduling is to select the maximum sum rate $\mu_{s,n}$ over $s$ and $n$ in (\ref{eq:NF-8}), implying that the maximum $\mu_{s,n}$ does not necessarily correspond to the sum of the maxima of $\gamma_{k}(s,n',q)$'s over $k$ and $n'$ for $q \in \Mc_{s,n}$. 
Nevertheless, we assume that the scheduler selects the maximum $\gamma_{k}(s,n',q)$ for simplicity. Otherwise, it is very difficult to analyze the throughput of the INR-based scheme by using extreme value theory since the standard approach requires the calculation of the distribution of $\sum_{q \in \Mc_{s,n}} \log\big(1+\max_{k}\gamma_{k}(s,n,q)\big)$, which is i.i.d. over $n$ for $s$ fixed but $\log\big(1+\max_{k}\gamma_{k}(s,n,q)\big)$'s are not independent over $q$. Suppose that the maximum $\mu_{s^*,n^*}$  corresponds to the sum of the $i_q$th largest $\gamma_{k}(s^*,n',q)$'s over $k$ and $n'$ for each $q \in \Mc_{s^*,n^*}$. Letting $$\underline{i} = \max_q i_q$$ we notice that the asymptotic behaviors of the largest $\gamma_{k}(s^*,n',q)$ and the $\underline{i}$th largest one would be quite similar for all $q$ since $K\times C^{M-1}_{s-1}$ is sufficiently large relative to $\underline{i}$ with high probability, i.e., the former quantity grows much faster than the latter as $M$ increases.\footnote{Recall that we consider the $M\sim K$ case with $M$ not small. Accordingly, $s^*$ near either $M$ or $1$ will almost surely not happen.}
Using this approximation, for large $K\times C^{M-1}_{s-1}$, we can write the throughput of INR as
\begin{align} \label{eq:TA-2}
   \mathcal{R}_\text{INR} &= \mathbb{E} \Big[\max_{s,n} \; \mu_{s,n} \Big] +o(1)\nonumber\\ 
   &\approx \mathbb{E} \bigg[\max_{s,n} \sum_{q \in \Mc_{s,n}} \log\Big(1+\max_{k}\gamma_{k}(s,n,q)\Big)\bigg] \nonumber\\ 
   &\approx \mathbb{E} \bigg[ \max_{s}\; s \log\Big(1+\max_{k,n'}\gamma_{k}(s,n',q)\Big)\bigg] \nonumber\\ 
   &= \max_{s}\mathbb{E} \bigg[ s \log\Big(1+\max_{k,n'}\gamma_{k}(s,n',q)\Big)\bigg] 
\end{align}
where $s=1,\cdots,M$, $n'=1,\cdots, C^{M-1}_{s-1}$, and $k=1,\cdots,K$.  The first approximation is analogous to (\ref{eq:TA-1}) and the second approximation follows from the equivalence of the asymptotic behaviors of the largest $\gamma_{k}(s,n',q)$ and the $\underline{i}$th largest one for given $s$ and $q$.
The last equality is due to the fact that $\gamma_{k}(s,n',q)$ is i.i.d. over $n'$ as well as $k$. Consequently, the flexible scheduling in the i.i.d. Rayleigh fading channel is to select the maximum of $K\times C^{M-1}_{s-1} $ i.i.d. random variables $\gamma_{k}(s,n',q)$ over both $k$ and $n'$ for each $s$ and $q$. Then, it chooses the optimal $s$ that maximizes (\ref{eq:TA-2}). Eventually, we obtain a simple result on the sum-rate scaling of INR as follows. 

\begin{thm}\label{thm-1}
 For $M$ finite, the achievable throughput of the INR scheme behaves in the i.i.d. Rayleigh fading channel as
\begin{align} \label{eq:TA-4}
   \mathcal{R}_\text{\emph{INR}} = \max_{1\le s\le M} s\log\log \big(K\times C^{M-1}_{s-1}\big) +s\log \frac{P}{s} +o(1).
\end{align}
\end{thm}

\begin{proof}
The proof of (\ref{eq:TA-4}) follows the same line of \cite{Sha05}. For a given $s$,  using the fact that $\gamma_{k}(s,n',q)$ is i.i.d. over both $k$ and $n'$, we recall that the $\underline{i}$th largest $\gamma_{k}(s,n',q)$ behaves like the largest one for $K\times C^{M-1}_{s-1}$ sufficiently large, where $\underline{i} \ll K\times C^{M-1}_{s-1}$. Then, applying extreme value theory \cite{Lea88}, it turns out that the maximum of $\gamma_{k}(s,n',q)$ over $(k, n')$ behaves like $\frac{P}{s}\log (K\times C^{M-1}_{s-1})$, thereby yielding (\ref{eq:TA-4}). A notable difference from the proof of RBF lies in that $K$ need not be necessarily large to invoke extreme value theory, unless $M$ is small. 
\end{proof}


The above result shows that, letting $S = s^*$, the flexible scheduling algorithm due to the INR feedback provides the multiuser diversity gain of $S\log\log \big(K\times C^{M-1}_{s-1}\big)$ and SNR gain of $\frac{M}{S}$ (i.e., SNR is $\frac{P}{S}$ instead of $\frac{P}{M}$) at the cost of multiplexing gain reduction by a factor of $\frac{S}{M}$, compared to (\ref{eq:TA-1b}) of RBF that requires sufficiently large $K$.
Note that this result holds true as long as $K$ is too small, as mentioned earlier. Therefore, Theorem \ref{thm-1} shows that flexible scheduling is of use for our goal that improves the conventional CB-based approach like RBF in the $M\ge K$ regime with $M$ large. Also, it is easy to see that the flexible scheduling gain vanishes as $K$ goes to infinity with $M$ finite, i.e., $\lim_{K\rightarrow \infty}\mathcal{R}_\text{INR} = \mathcal{R}_\text{RBF}$ with $S=M.$

\begin{figure} 
\hspace{-4mm} \includegraphics[scale=.62]{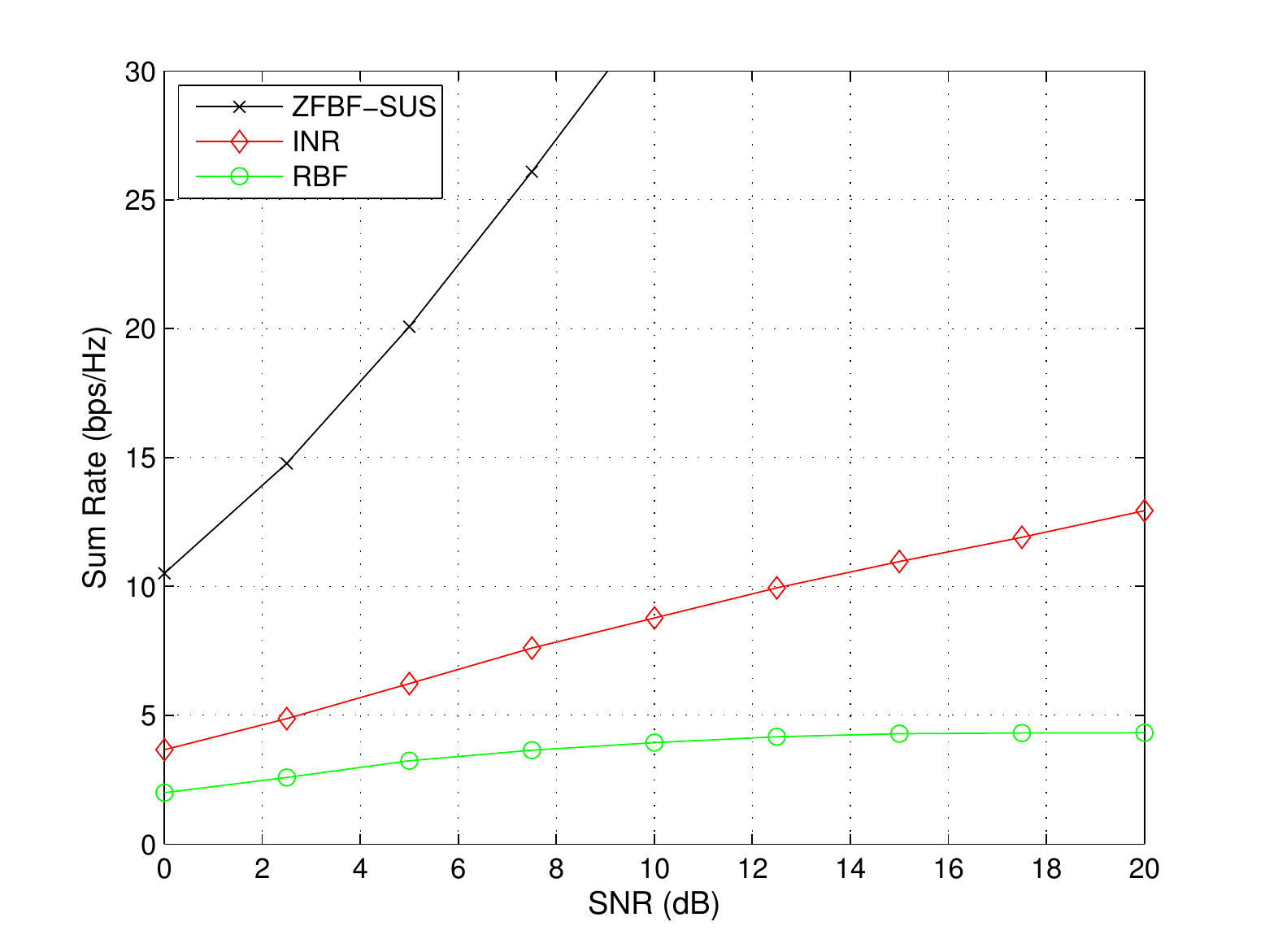} 
\caption{Achievable sum rate curves of MU-MIMO downlink schemes versus SNR in the independent fading case for $M=16$ and $K=20$.} \label{fig-1}
\end{figure} 

Fig. \ref{fig-1} shows the achievable throughputs of ZFBF-SUS \cite{Yoo06}, RBF, and INR schemes in the i.i.d. Rayleigh fading channel for $M=16$ and $K=20$. We can see that the rate gap between RBF and INR is quite large, advocating the flexible scheduling gain is remarkable. Even if INR outperforms RBF, it suffers severe performance degradation compared to ZFBF in the \emph{independent fading} case.

Extending Theorem \ref{thm-1}, where we assumed $T=1$, into the multiple $T$ case, we have the following corollary.

\begin{corol}\label{cor-1}
 For $M$ finite, the achievable throughput of the full INR feedback scheme behaves in the i.i.d. Rayleigh fading channel as
\begin{align} \label{eq:TA-6}
   \mathcal{R}_\text{\emph{full INR}} = \max_{1\le s\le M} s\log\log \Big(KT\times C^{M-1}_{s-1}\Big) +s\log \frac{P}{s} +o(1).
\end{align}
Assuming the same number of users for each subset $\Cc_t$ in (\ref{eq:Pre-3b}), the achievable throughput of the partial INR feedback scheme behaves like
\begin{align} \label{eq:TA-7}
   \mathcal{R}_\text{\emph{partial INR}} = \max_{1\le s\le M} s\log\log \Big(\frac{K}{T}\times C^{M-1}_{s-1}\Big) +s\log \frac{P}{s} +o(1).
\end{align}
\end{corol}

Comparing (\ref{eq:TA-6}) with (\ref{eq:TA-7}), we can see that the full INR feedback gives only a vanishing gain over the partial INR as $K$ increases for $T$ fixed.



\section{Feedback Overhead Reduction}

One may argue that even the $2M$-bit INR feedback per user is still burdensome for large-scale MIMO, where both $M$ and $K$ are large. 
We present here a simple and heuristic approach to reduce the cost of INR feedback, referred to as \emph{one-bit INR feedback}.

For each user $k$, first select $(t_k,m_k)$ such that 
\begin{align} \label{eq:OR-1}
   (t_k,m_k) = \argmax_{1\le t\le T, 1\le m\le M} \frac{\big|\ph_{k}^\ct\pc_m^{(t)}\big|^2}{\sigma_{k}^2}.
\end{align}
We consider the resulting ${|\ph_{k}^\ct\pc_{m_k}^{(t_k)}|^2}/{\sigma_{k}^2}$ as the SNR (denoted by $\snr_k$) of user $k$ with respect to the selected beam $(t_k,m_k)$ and hence the partial INR feedback is assumed here. 
Then, we impose a threshold $\gamma$ on $\inr_{k,j}^{(t_k)}, j\neq m_k$ in (\ref{eq:NF-4}). If 
$$\frac{\inr_{k,j}^{(t_k)}}{\snr_k}\ge \gamma$$
then one bit is assigned to beam $j$ by `1' and otherwise by `0'. By setting $\gamma$ sufficiently small (e.g., $\gamma = 0.01$ ($-20$ dB)), this feedback information allows the BS to just know whether the interference between any selected users is negligible or not, instead of the accurate INR. Since at high SNR the receiver performance is more sensitive to the channel estimation error, $\gamma$ slightly depends on $\snr_k$. Therefore, $\gamma$ may be a function of $\snr_k$ rather than a fixed quantity.

Clearly, the one-bit INR feedback is to impose a restriction on the scheduler by reducing the number of hypotheses. At the sacrifice of an inevitable performance loss, this technique provides feedback overhead reduction by a factor of $1/2$ as it consumes only one SNR (e.g., $4$ bits) and $(M-1)$-bit INRs per user. The performance loss will be evaluated later by simulation results.

[OPEN ISSUES]

1) We definitely need a low complexity scheduling algorithm for the one-bit INR feedback scheme, which may be  computationally much more efficient than others for large-scale MIMO. 

2)
Note that `1' would be sparse when the AS of users are large. 
On the other hand, `0' would become more sparse as the number $M$ of antennas and/or the center frequency $f_c$ of the system increases. To exploit these aspects and hence further reduce the feedback overhead, we need to derive or at least employ a compressive sensing technique for binary sparse signals.

\section{Numerical Results}

We evaluate achievable throughputs of several schemes of interest in correlated fading MIMO broadcast channels, for which we use the following one-ring channel model.

\subsection{One-Ring Channel model}

In the typical cellular
downlink case, the BS is elevated and free of local
scatterers, and the users are placed at ground level
and are surrounded by local scatterers. This channel scenario corresponds to the one-ring model \cite{Shi00}, for which a user located at azimuth angle $\theta$ and distance ${\sf s}$ is surrounded by a ring of scatterers of radius  ${\sf r}$ such that AS $\Delta \approx \arctan({\sf r}/{\sf s})$. Assuming the uniform linear array (ULA) with a uniform distribution of the received power from planar waves impinging on the BS array, we have that the correlation coefficient
between BS antennas $1 \leq p, q \leq M$ is given by \cite{Adh13}
\begin{align} \label{eq:SM-4} 
   [\pR]_{p,q} =\frac{1}{2\Delta}  \int_{-\Delta}^{\Delta}e^{j2\pi D(p-q)\sin(\alpha+\theta)}d\alpha
\end{align}
where $D$ is the normalized distance between antenna elements by the wavelength. 

\subsection{System Throughput}

We will present here the achievable throughput of MU-MIMO downlink systems discussed in this paper. Consider the BS has $M=4,8,16$ antennas with ULA of antenna spacing $D=1/2$, and users are uniformly distributed over the range $[-60^o, 60^o]$ with $\Delta_k$ uniformly distributed in the range $[5^o,20^o]$, where $\Delta_k$ is angular spread (AS) of user $i$. The transmit correlation matrices are generated by the one-ring channel model (\ref{eq:SM-4}).

\subsubsection{Flexible vs. Limited Scheduling} 

\begin{figure} 
\hspace{-4mm} \includegraphics[scale=.62]{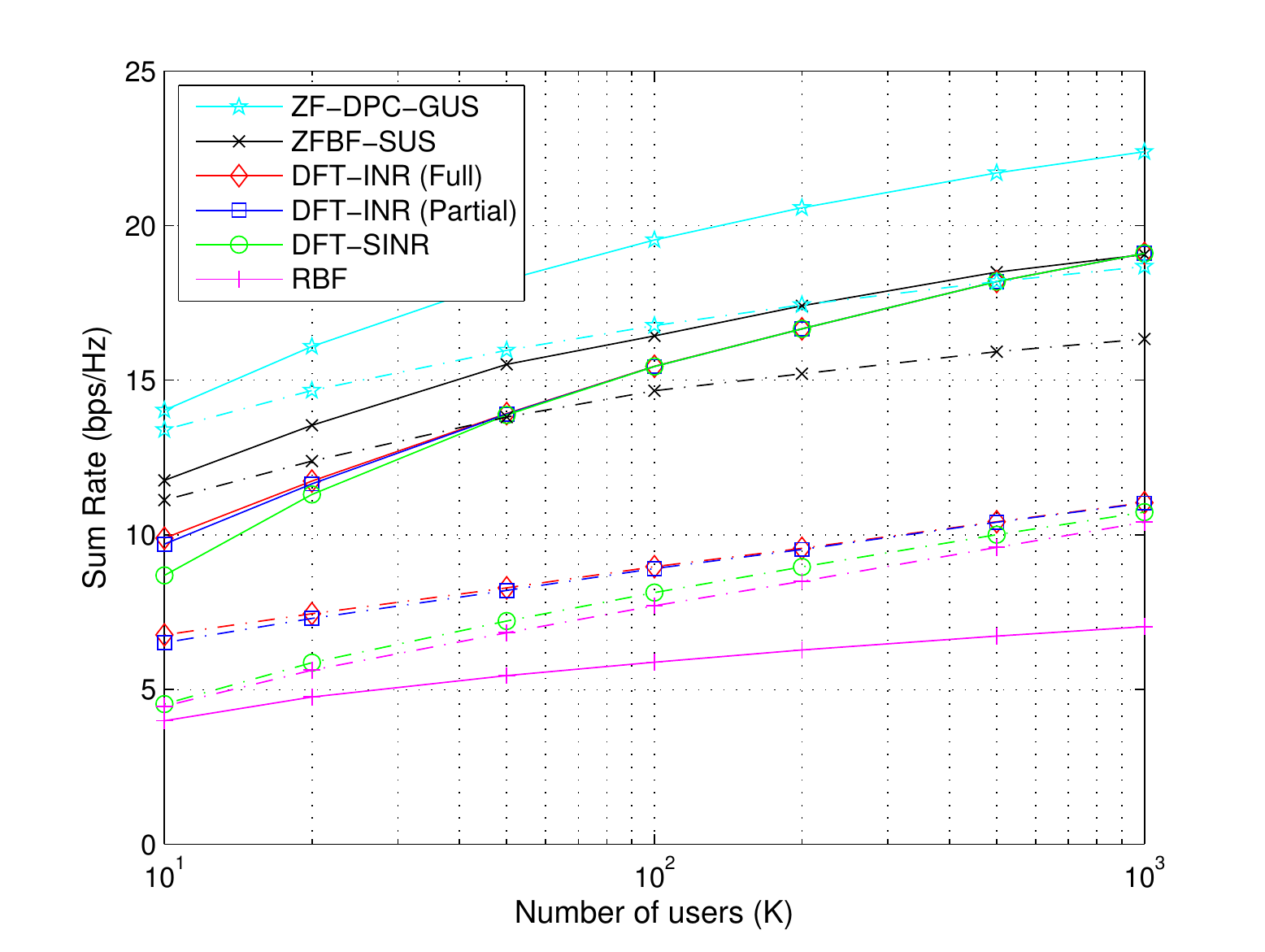} 
\caption{Achievable sum rate curves of various MU-MIMO downlink schemes versus $K$ for $M=4$,  SNR $=10$ dB, and $\Delta_k\in[5^o,20^o]$. While the solid lines indicate spatially correlated Rayleigh fading channels, the dash-dot lines denote the i.i.d. Rayleigh fading channel.} \label{fig-2}
\end{figure} 

For $M=4$ (i.e., small $M$), Fig. \ref{fig-2} compares the sum rates of ZF-DPC with greedy user selection \cite{Dim05} (ZF-DPC-GUS), which approximates the sum capacity, ZFBF-SUS, DFT-based precoding with  full INR feedback (DFT-INR-Full), DFT-based precoding with partial INR feedback (DFT-INR-Partial), DFT-based precoding with SINR feedback (DFT-SINR), random beamforming (RBF) \cite{Sha05}. While ZF-DPC-GUS and ZFBF-SUS assume perfect CSIT, the others employ the DFT-based CB (\ref{eq:Pre-3}). 
Throughout this section, the number $T$ of subsets in the CB was set to $2$, which is appropriate as a large CB size usually undermines multiuser diversity gain. The beneficial impact of transmit correlation on the capacity of MU-MIMO downlink channels was addressed in \cite{Nam14a}. We show that spatially correlated fading is more beneficial to the CB-based schemes of interest than ZFBF with full CSI. 
Fig. \ref{fig-3} and Fig. \ref{fig-4} present the sum-rate comparison for $M=8,16$ (slightly large $M$),  respectively. 



\begin{figure} 
\hspace{-4mm} \includegraphics[scale=.62]{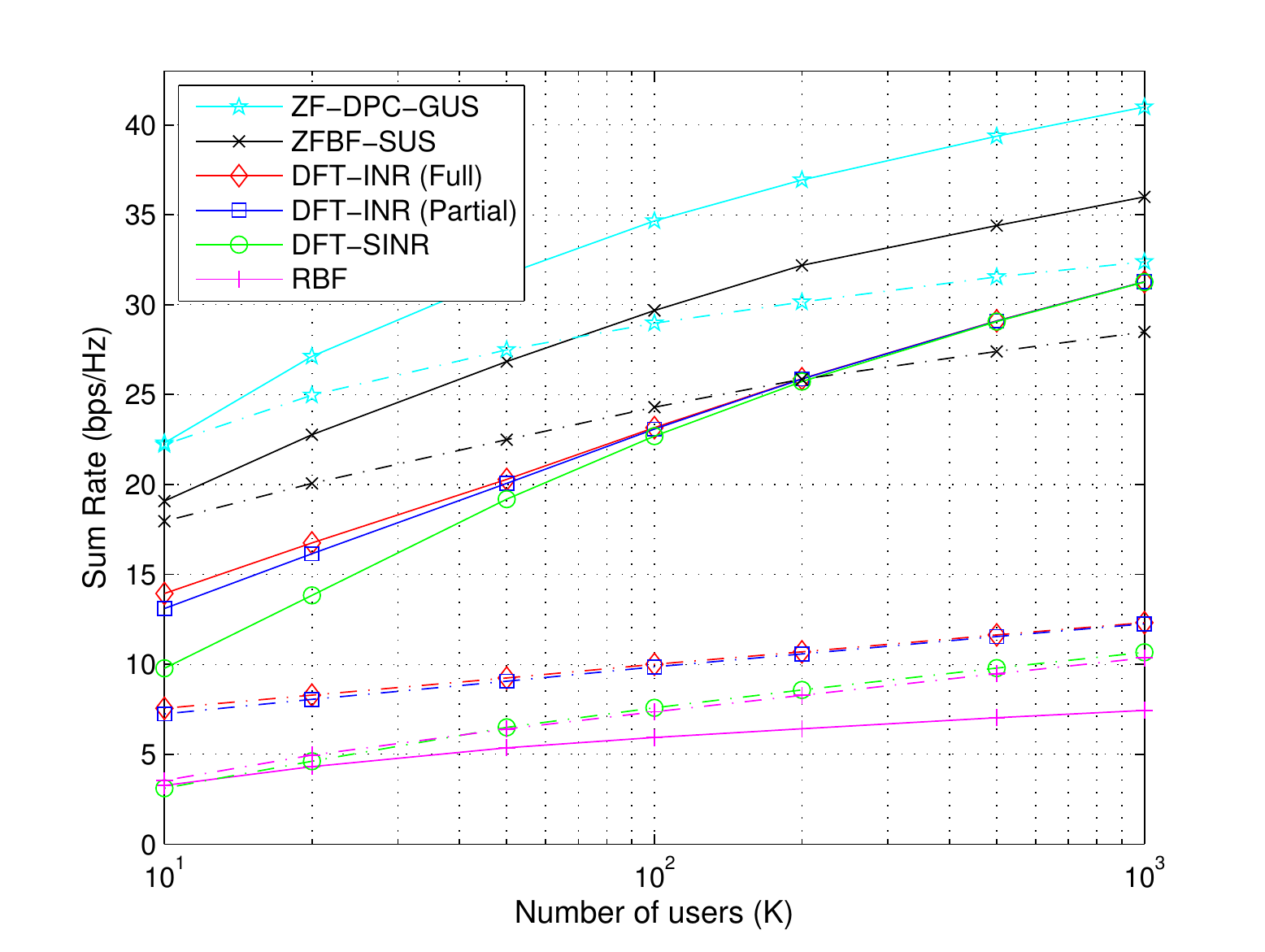} 
\caption{Achievable sum rate curves of various MU-MIMO schemes versus $K$ for $M=8$,  SNR $=10$ dB, and $\Delta_k\in[5^o,20^o]$. While the solid lines indicate spatially correlated Rayleigh fading channels, the dash-dot lines denote the i.i.d. Rayleigh fading channel.} \label{fig-3}
\end{figure} 

\begin{figure} 
\hspace{-4mm} \includegraphics[scale=.62]{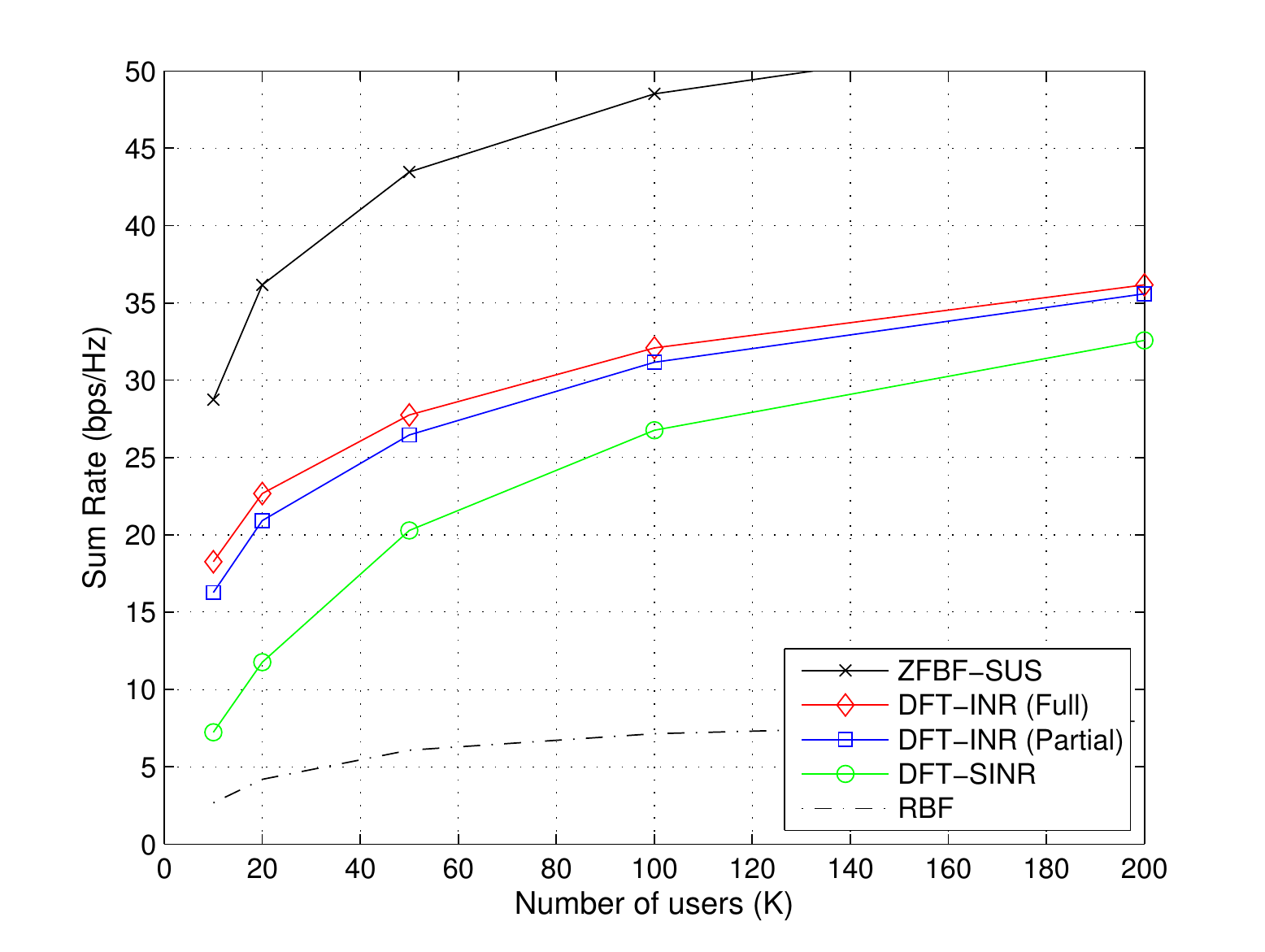} 
\caption{Achievable sum rate curves of various MU-MIMO schemes versus $K$ for $M=16$,  SNR $=10$ dB, and $\Delta_k\in[5^o,20^o]$.} \label{fig-4}
\end{figure} 

From Figs. \ref{fig-2}--\ref{fig-4}, we observe that the proposed framework based on INR feedback improves the system throughput due to flexible scheduling, compared to the conventional framework based on CDI and SINR feedback (i.e, DFT-SINR), especially when $K$ is small. 
Therefore, the flexible scheduling gain of the new framework is readily expected to become even much larger as $M$ increases at the cost of a moderate increase in CSI feedback overhead.
It is also shown that DFT-SINR approaches DFT-INR for $K$ sufficiently large. This is because the optimal scheduling selects all $M$ orthogonal beams in a subset $t$ with high probability as $K\rightarrow \infty$, which corresponds to the SINR feedback case.

Finally, we observe that the DFT-based schemes significantly outperform RBF in correlated fading channels, the latter of which is known to be quite inferior to ZFBF with limited feedback (e.g., \cite{Rav08}). The reason why RBF shows severe degradation in the correlated fading channel is that it uses isotropic orthonormal beams even though the user channels have long-term channel directions, i.e., some of beams transmitted by RBF are outside the channel directions of users.



\subsubsection{Impact of Transmit Correlation}

\begin{figure} 
\hspace{-4mm} \includegraphics[scale=.62]{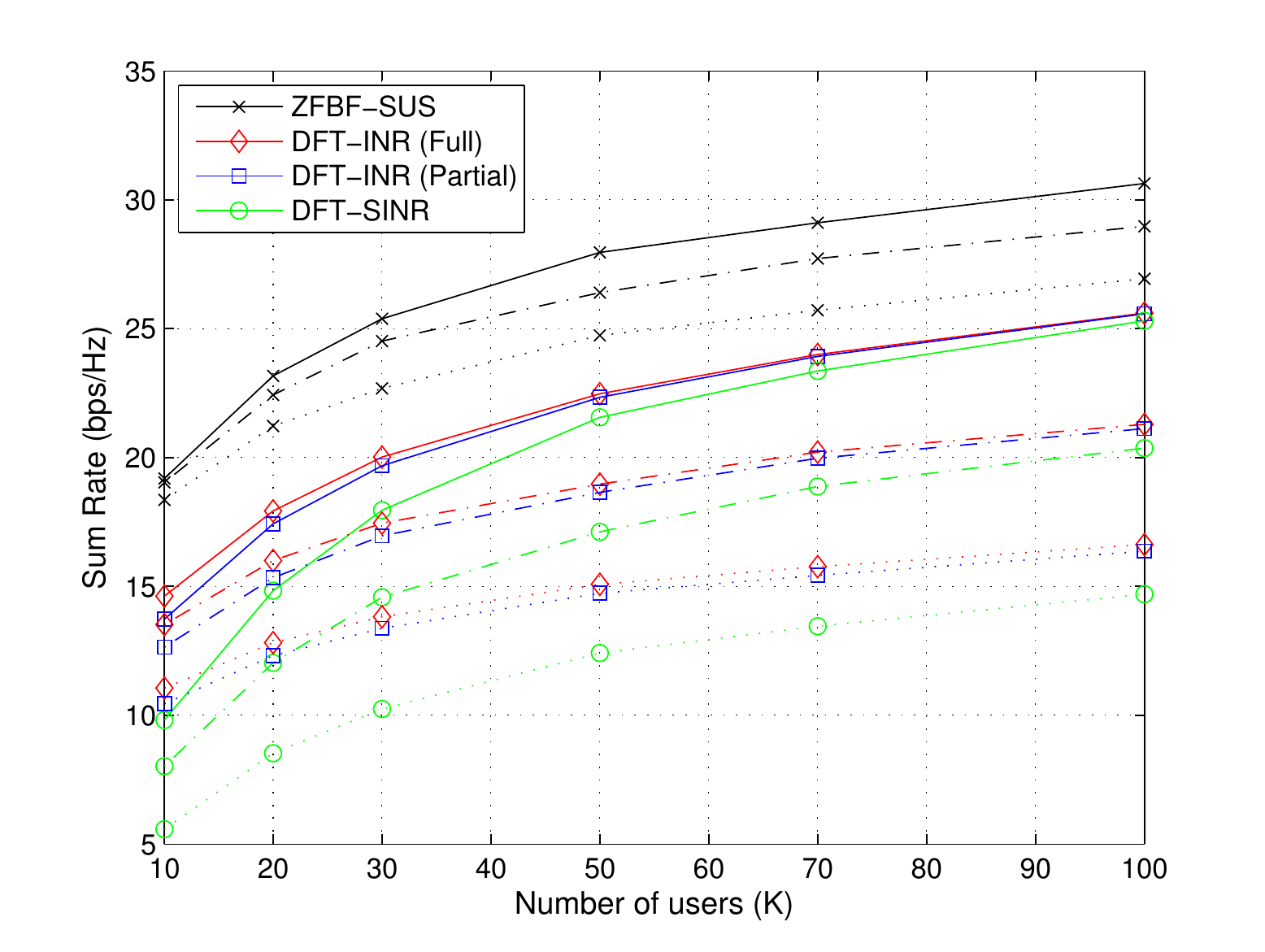} 
\caption{Impact of transmit correlation on sum rates of various MU-MIMO  schemes for varying $K$ where $M=8$ and SNR $=10$ dB. The solid, dash-dot, and dotted lines indicate $\Delta_k\in[5^o,10^o]$, $\Delta_k\in[10^o,20^o]$, and $\Delta_k\in[20^o,40^o]$, respectively.} \label{fig-5}
\end{figure}

Fig. \ref{fig-5} shows the effect of three different transmit correlation scenarios captured by the distribution of $\Delta_k$ on achievable rates of ZFBF-SUS, DFT-INR, and DFT-SINR for $M=8$. As transmit correlation increases (i.e., $\Delta_k$ decreases), the rate gap between ZFBF-SUS and DFT-INR reduces, thus implying DFT-INR lends itself to highly correlated fading channels.

\subsubsection{Sensitivity to Noisy CSIT}

\begin{figure} 
\hspace{-4mm} \includegraphics[scale=.62]{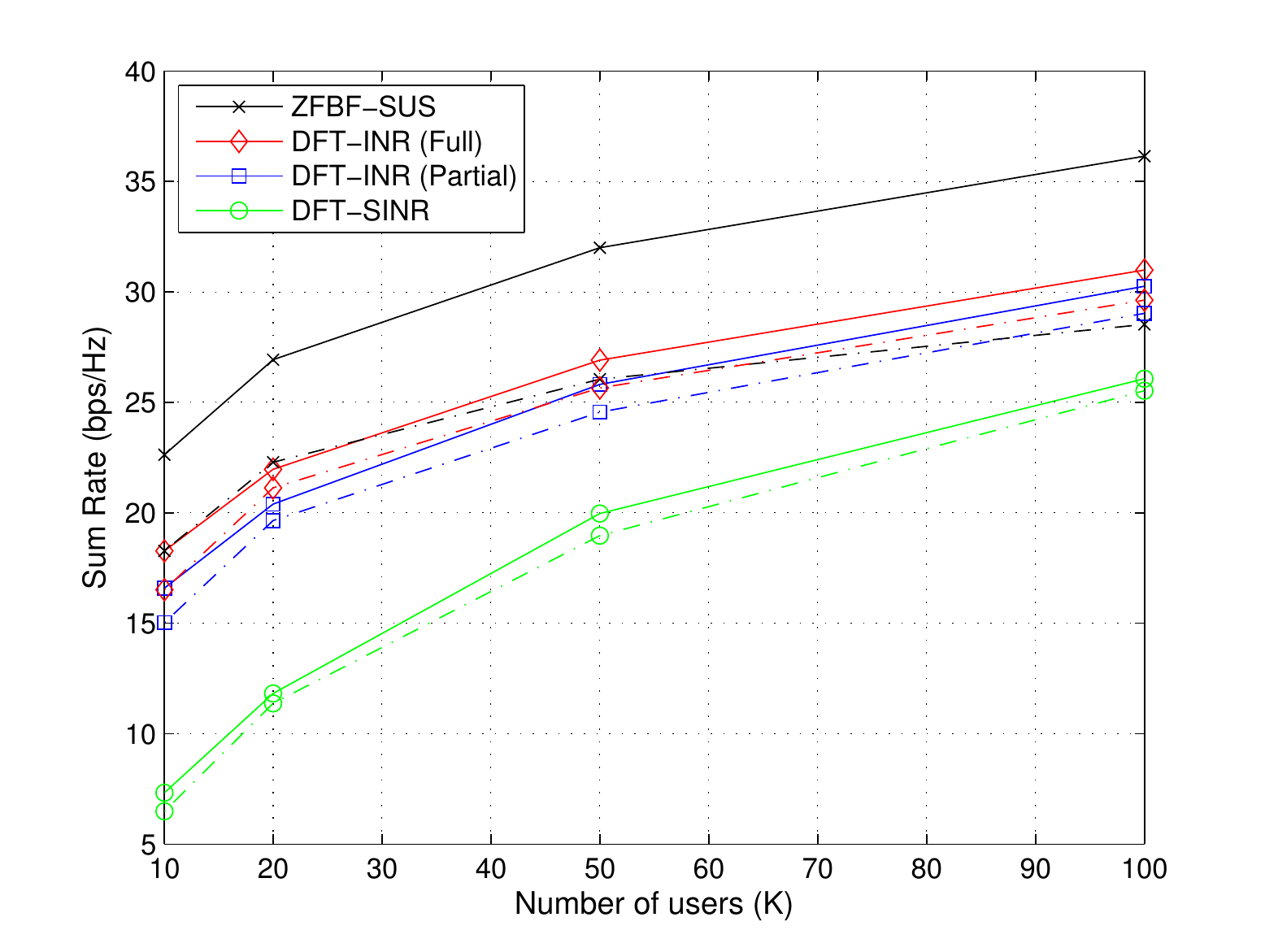} 
\caption{Impact of channel estimation error on sum rates of various MU-MIMO schemes for varying $K$ where $M=16$,  SNR $=10$ dB, and $\Delta_k\in[5^o,20^o]$. The solid lines and dash-dot lines indicate $\sigma_\text{err}^2 =0.1$ and $0.2$, respectively.} \label{fig-6}
\end{figure}

Figure \ref{fig-6} shows the sensitivity of ZFBF-SUS, DFT-INR, and DFT-SINR to noisy CSIT $\tilde{\ph}_k$, which is modeled as $\tilde{\ph}_k =\sqrt{1-\sigma_\text{err}^2}\;\ph_k + \sigma_\text{err} \pn$, where $\pn \sim\mathcal{CN} (\p0, \pI)$ and $0 \le \sigma_\text{err}^2 <1$.
As readily expected, ZFBF-SUS with noisy CSIT (i.e., channel estimation error due to imperfect CSIR) suffers relatively large performance degradation, while DFT-INR is less sensitive. 
Notice that this channel estimation error incurs additional performance degradation, aside from channel quantization error.

\subsubsection{Impact of Dedicated Pilot Overhead}

\begin{figure} 
\hspace{-4mm} \includegraphics[scale=.62]{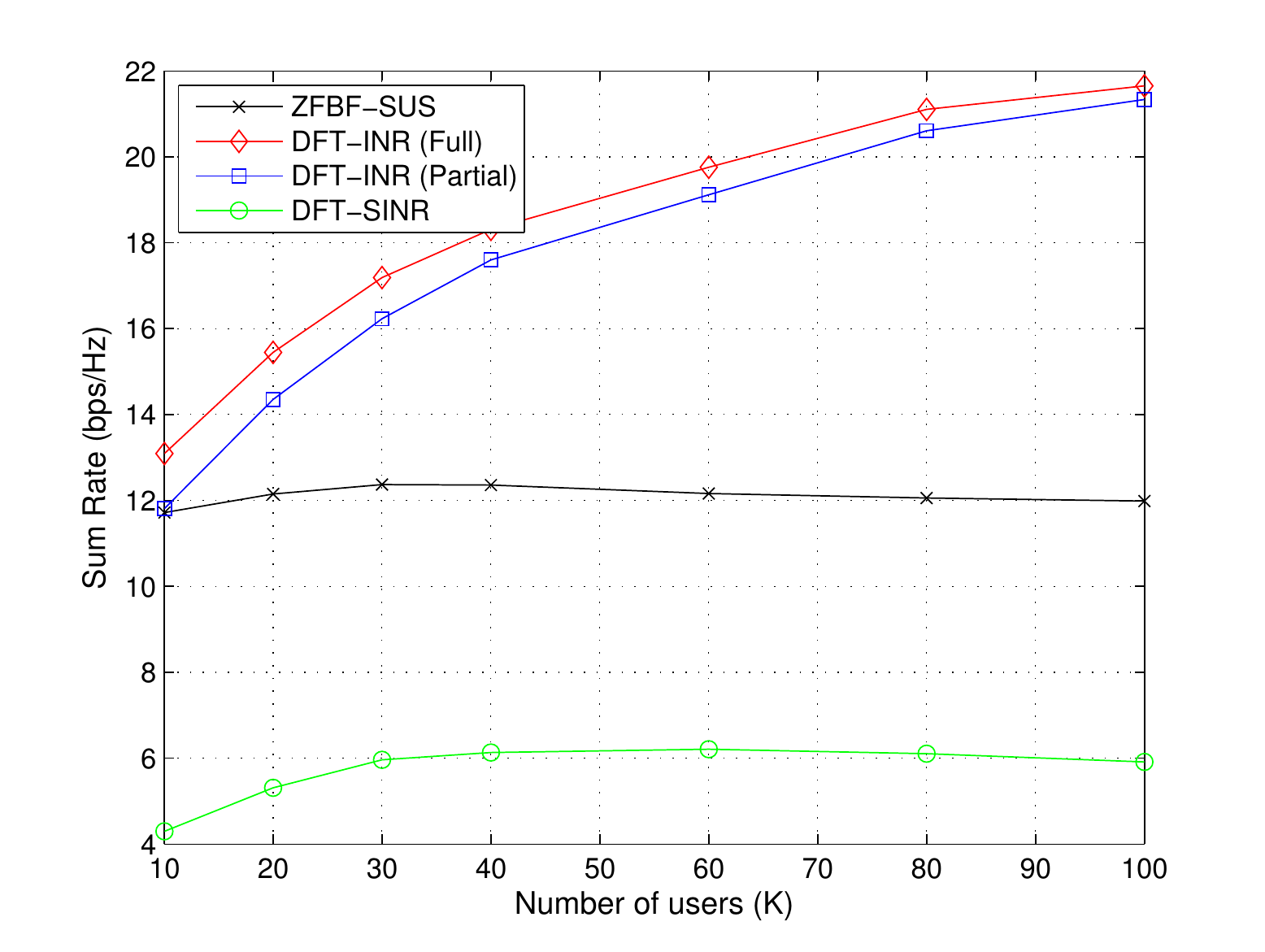} 
\caption{Impact of dedicated pilot overhead on sum rates of various MU-MIMO schemes for varying $K$ where $M=16$,  SNR $=10$ dB, $\Delta_k\in[5^o,20^o]$, and $\sigma_\text{err}^2 =0.1$.} \label{fig-7}
\end{figure}

So far, we have assumed the scheduler can use all the resources during the communication phase. However, we should consider a cost for the downlink dedicated pilot phase in realistic systems, regardless of FDD or TDD.
For convenience, we use some numerology of LTE standards to study the overhead for dedicated pilots in a more realistic fashion. Given a resource block (RB) consisting of 14 OFDM symbols with 12 subcarriers each, the LTE-advanced system allocates roughly 3 symbols for control channel and others (CSI-RS, etc) and 1 symbol per 2 data streams (DM-RS ports) for MU-MIMO downlink. For the conventional framework with SINR feedback (including ZFBF-SUS with noisy CSIT), the ratio of the remaining OFDM symbols to convey user data streams to the total symbols can be expressed as 
\begin{align} \label{eq:NR-2}
   \kappa_\text{no-INR}(i) = \frac{11-\lceil \frac{S(i)}{2}\rceil}{14}
\end{align}
where $S(i)\le \min(M,K)$ is the number of selected users at the $i$th scheduling time interval.
In contrast, we let the ratio $\kappa_\text{INR}$ of DFT-INR fixed to $10/14$. Say, the BS schedules up to $S(i)$ users only when one group of $S(i)$ users can share a common resource for their dedicated pilots with negligible interference to each other (e.g., $-20$ dB lower than each intended signal power) and the other group does so, resulting in that two user groups consume only a single symbol. 


The curves in Fig. \ref{fig-7} show the adjusted system throughputs of the schemes of interest for $M=16$ with noisy CSIT ($\sigma_\text{err}^2=0.1$), where the throughputs of ZFBF-SUS and DFT-SINR are adjusted by the scaling factor $\kappa_\text{no-INR}(i)$ every scheduling interval, while DFT-INR is scaled by $\kappa_\text{INR}=10/14$. This result points out that the dedicated pilot overhead significantly affects the system performance. The performance of DFT-INR scales with $K$, while the others do not. Noticing that $\kappa_\text{no-INR}(i)$ keeps growing as both $M$ and $K$ increase, it is  expected that DFT-INR will show larger performance benefits when $M$ is moderately large (e.g., up to $100$).

\begin{figure} 
\hspace{-4mm} \includegraphics[scale=.62]{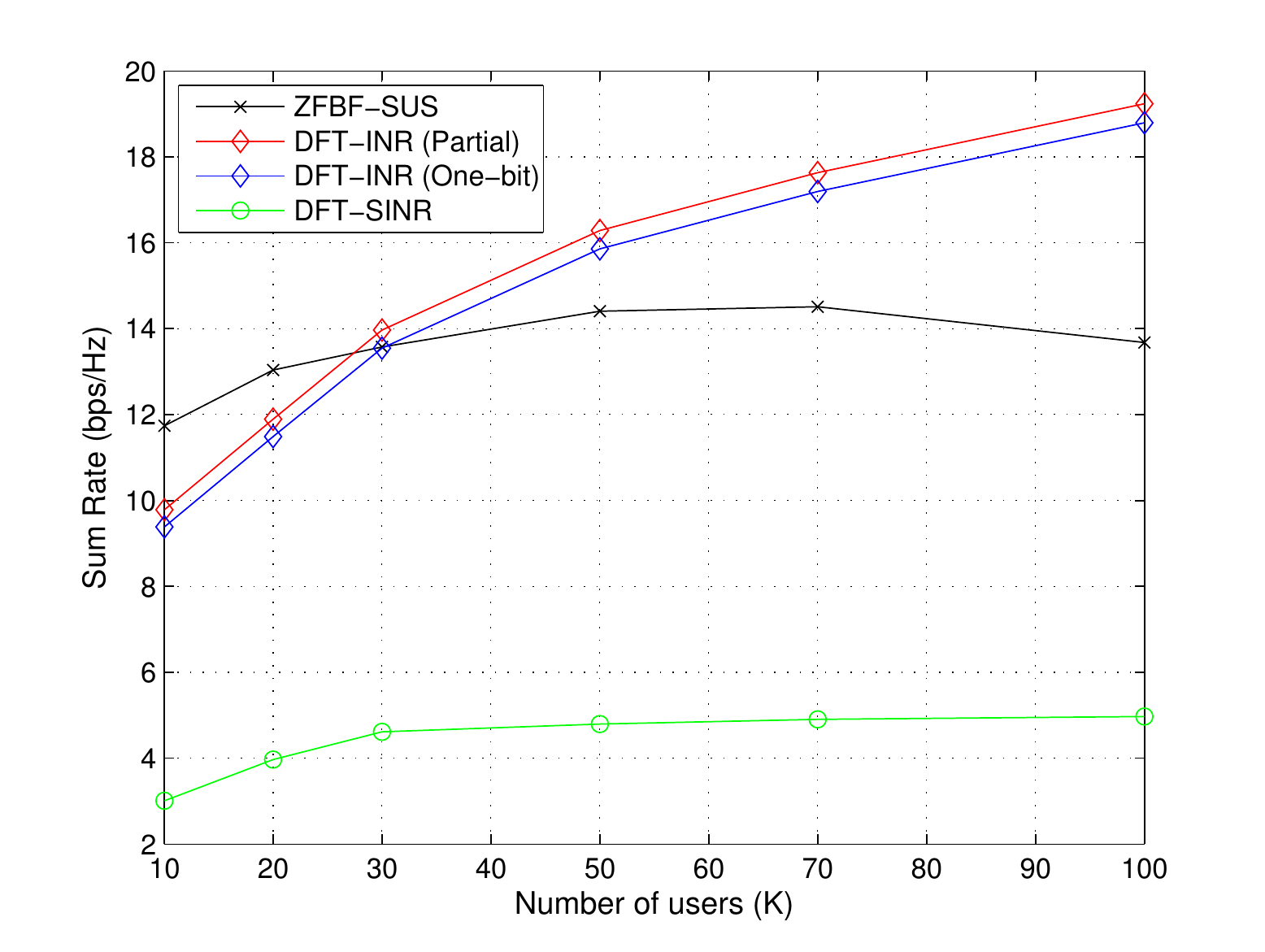} 
\caption{System throughput of the one-bit INR feedback scheme for $M=16$ with 3GPP case 1 parameters.} \label{fig-8}
\end{figure}

Finally, Fig. \ref{fig-8} shows the average system throughput of the one-bit INR feedback scheme with some large-scale channel parameters like the mean of $\Delta_k$ and path loss exponent (heterogeneous SNRs among users) obeying 3GPP case 1 (urban macro scenario) \cite{SCM} with $\sigma_\text{err}^2 =0.1$ and $\gamma =0.02$ ($-17$ dB). Compared to the full INR feedback, the one-bit INR exhibits remarkably marginal degradation in the more realistic situation.

\section{Conclusion}


We have proposed a new framework for CB-based MU-MIMO systems, which requires no explicit feedback of CDI and rather relies on a new type of CQI feedback referred to as INR.
The benefits of the new framework can be summarized as flexible scheduling, dedicated pilot overhead reduction, and arbitrary power allocation.

Assuming a CB of size $M$, limited scheduling of RBF selects the largest one of $K$ SINRs per beam, while flexible scheduling due to INR is asymptotically equivalent to select the largest one of $K\times C^{M-1}_{S-1}$ SINRs per beam, where $S$ is the number of selected users. In the i.i.d. Rayleigh fading channel, the resulting multiuser diversity gain is $S\log\log (K\times C^{M-1}_{S-1})$ for large $K\times C^{M-1}_{S-1}$ but not requiring large $K$. In contrast, the well-known $M\log\log K$ of RBF is valid only when $K$ is sufficiently large with $K\gg M$. Therefore, if $K\sim M$ with $M$ not small, the performance improvement thanks to INR feedback was shown to be considerable. 
It is arguably most interesting and surprising that taking the cost of downlink dedicated pilot into account, DFT-INR under the new limited feedback framework may significantly outperform the ideal ZFBF-SUS without CDI quantization even for the large-scale array regime.

%
%
%
%

\bibliographystyle{IEEEtran}
\bibliography{INR_ArXiv}

\end{document}